\documentclass[11pt]{amsart}
\usepackage{url}
\usepackage{xcolor}
\usepackage{graphicx} 
\usepackage{color}
\usepackage{hyperref}
\definecolor{refcolor}{rgb}{0.1,0.2,0.88}
\hypersetup{
    colorlinks,
    citecolor=refcolor,
    filecolor=refcolor,
    linkcolor=refcolor,
    urlcolor=refcolor
}

\theoremstyle{definition}
\newtheorem{theorem}{Theorem}[section]

\newtheorem{remark}[theorem]{Remark}

\def\({\left(}
\def\){\right)}

\newcommand{\R}{\mathbb{R}}

\newcommand{\de}{\textnormal{d}}

\newcommand{\vol}{\de_{vol}}
\newcommand{\volspace}{\de_{vol_3}}

\newcommand{\tn}{\textnormal}
\newcommand{\ds}{\displaystyle}

\newcommand{\vs}{\textit{vs.} }

\newcommand{\eg}{\textit{e.g.} }

\newcommand{\citep}[2]{\cite{#1}, p. #2}

\newcommand{\mc}[1]{\mathcal{#1}}

\newcommand{\sref}[1]{\S\ref{#1}}

\newcommand{\image}[3]{
\begin{figure}[ht]
\begin{center}
\includegraphics[width=#2\textwidth]{#1}
\caption{\small{\label{#1}#3}}
\end{center}
\end{figure}
}

\newcommand{\dsfrac}[2]{\ds{\frac{#1}{#2}}}

\newcommand{\flrw}{Friedmann-Lema\^itre-Robertson-Walker}
\newcommand{\FLRW}{FLRW}

\def\hyph{-\penalty0\hskip0pt\relax}
\newcommand{\semiriem}{semi{\hyph}Riemannian}

\newcommand{\semireg}{semi{\hyph}regular}
\newcommand{\quasireg}{quasi{\hyph}regular}

\newcommand{\nondeg}{non{\hyph}degenerate}

\begin{document} 
 
\title[The FLRW Big Bang singularities are well behaved]{The Friedmann-Lema\^itre-Robertson-Walker Big Bang singularities are well behaved}

\author{Ovidiu Cristinel \ Stoica}
\date{\today. Horia Hulubei National Institute for Physics and Nuclear Engineering, Bucharest, Romania. E-mail: cristi.stoica@theory.nipne.ro, holotronix@gmail.com}

\begin{abstract}
We show that the Big Bang singularity of the Friedmann-Lema\^itre-Robertson-Walker model does not raise major problems to General Relativity. We prove a theorem showing that the Einstein equation can be written in a non-singular form, which allows the extension of the spacetime before the Big Bang. The physical interpretation of the fields used is discussed.

These results follow from our research on singular semi-Riemannian geometry and singular General Relativity.
\bigskip
\noindent 
\end{abstract}


\maketitle

\setcounter{tocdepth}{1}
\tableofcontents

\section{Introduction}

\subsection{The universe}

According to the {\em cosmological principle}, our expanding universe, although it is so complex, can be considered at very large scale homogeneous and isotropic. This is why we can model the universe, at very large scale, by the solution proposed by A. Friedmann \cite{FRI22de,FRI99en,FRI24}. This exact solution to Einstein's equation, describing a homogeneous, isotropic universe, is in general called the {\flrw} (\FLRW) metric, due to the rediscovery and contributions made by Georges Lema\^itre \cite{LEM27}, H. P. Robertson \cite{ROB35I,ROB35II,ROB35III} and A. G. Walker \cite{WAL37}.

The {\FLRW} model shows that the universe should be, at a given moment of time, either in expansion, or in contraction. From Hubble's observations, we know that the universe is currently expanding. The {\FLRW} model shows that, long time ago, there was a very high concentration of matter, which exploded in what we call the {\em Big Bang}. Was the density of matter at the beginning of the universe so high that the Einstein's equation was singular at that moment? This question received an affirmative answer, under general hypotheses and considering General Relativity to be true, in Hawking's singularity theorem \cite{Haw66i,Haw66ii,Haw67iii} (which is an application of the reasoning of Penrose for the black hole singularities \cite{Pen65}, backwards in time to the past singularity of the Big Bang).

Given that the extreme conditions which were present at the Big Bang are very far from what our experience told us, and from what our theories managed to extrapolate up to this moment, we cannot know precisely what happened then. If because of some known or unknown quantum effect the energy condition from the hypothesis of the singularity theorem was not obeyed, the singularity might have been avoided, although the density was very high. One such possibility is explored in the {\em loop quantum cosmology} \cite{bojowald2003absenceLQC,Boj05,Ash08,ashtekar2011LQC,visinescu2009bianchi,visinescu2012bianchi}, which leads to a Big Bounce discrete model of the universe.

Not only quantum effects, but also classical ones, for example repulsive forces, can avoid the conditions of the singularity theorems, and prevent the occurrence of singularities. An important example comes from non-linear electrodynamics, which allows the construction of a stress-energy tensor which removes the singularities, as it is shown in \cite{corda2010removingBHsingularities} for black holes, and in \cite{corda2011inflation} for cosmological singularities.

We will not explore here the possibility that the Big Bang singularity is prevented to exist by quantum or other kind of effects, because we don't have the complete theory which is supposed to unify General Relativity and Quantum Theory. What we will do in the following is to push the limits of General Relativity to see what happens at the Big Bang singularity, in the context of the {\FLRW} model. We will see that the singularities are not a problem, even if we don't modify General Relativity and we don't assume very repulsive forces to prevent the singularity.

One tends in general to regard the singularities arising in General Relativity as an irremediable problem which forces us to abandon this successful theory \cite{HP70,ASH91,HP96,Ash08}. In fact, contrary to what is widely believed, we will see that the singularities of the {\FLRW} model are easy to understand and are not fatal to General Relativity. In \cite{Sto11a} we presented an approach to extend the {\semiriem} geometry to the case when the metric can become degenerate. In \cite{Sto11b} we applied this theory to the warped products, by this providing means to construct examples of singular {\semiriem} manifolds of this type. We will develop here some ideas suggested in some of the examples presented there, and apply them to the singularities in the {\FLRW} spacetime. We will see that the singularities of the {\FLRW} metric are even simpler than the black hole singularities, which we discussed in \cite{Sto11e,Sto11f,Sto12e}.

\subsection{The {\flrw} spacetime}
\label{s_flrw_intro}

Let's consider the $3$-space at any moment of time as being modeled, up to a scaling factor, by a three-dimensional Riemannian space $(\Sigma,g_\Sigma)$. The time is represented as an interval $I\subseteq \R$, with the natural metric $-\de t^2$. At each moment of time $t\in I$, the space $\Sigma_t$ is obtained by scaling $(\Sigma,g_\Sigma)$ with a scaling factor $a^2(t)$. The scaling factor is therefore given by a function $a: I\to \R$, named the {\em warping function}. The {\FLRW} spacetime is the manifold $I\times\Sigma$ endowed with the metric
\begin{equation}
\label{eq_flrw_metric}
\de s^2 = -\de t^2 + a^2(t)\de\Sigma^2,
\end{equation}
which is the {\em warped product} between the manifolds $(\Sigma,g_\Sigma)$ and $(I,-\de t^2)$, with the warping function $a: I\to \R$.

The typical space $\Sigma$ can be any Riemannian manifold we may need for our cosmological model, but because of the homogeneity and isotropy conditions, it is in general taken to be, at least at large scale, one of the homogeneous spaces $S^3$, $\R^3$, and $H^3$. In this case, the metric on $\Sigma$ is, in spherical coordinates $(r,\theta,\phi)$,
\begin{equation}
\label{eq_flrw_sigma_metric}
\de\Sigma^2 = \dsfrac{\de r^2}{1-k r^2} + r^2\(\de\theta^2 + \sin^2\theta\de\phi^2\),
\end{equation}
where $k=1$ for the $3$-sphere $S^3$, $k=0$ for the Euclidean space $\R^3$, and $k=-1$ for the hyperbolic space $H^3$.

\subsection{The Friedmann equations}

Once we choose the $3$-space $\Sigma$, the only unknown part of the {\FLRW} metric is the function $a(t)$. To determine it, we have to make some assumptions about the matter in the universe. In general it is assumed, for simplicity, that the universe is filled with a fluid with mass density $\rho(t)$ and pressure density $p(t)$. The density and the pressure are taken to depend on $t$ only, because we assume the universe to be homogeneous and isotropic. The stress-energy tensor is
\begin{equation}
\label{eq_friedmann_stress_energy}
T_{ab} = \(\rho+p\)u_a u_b + p g_{ab},
\end{equation}
where $u^a$ is the timelike vector field $\partial_t$, normalized.

From the energy density component of the Einstein equation, one can derive the {\em Friedmann equation}
\begin{equation}
\label{eq_friedmann_density}
\rho = \dsfrac{3}{\kappa}\dsfrac{\dot{a}^2 + k}{a^2},
\end{equation}
where $\kappa:=\dsfrac{8\pi \mc G}{c^4}$ ($\mc G$ and $c$ being the gravitational constant and the speed of light, which we will consider equal to $1$ for now on, by an appropriate choice of measurement units).
From the trace of the Einstein equation, we obtain the {\em acceleration equation}
\begin{equation}
\label{eq_acceleration}
\rho + 3p = -\dsfrac{6}{\kappa}\dsfrac{\ddot{a}}{a}.
\end{equation}
The {\em fluid equation} expresses the conservation of mass-energy:
\begin{equation}
\label{eq_fluid}
\dot{\rho} = -3 \dsfrac{\dot{a}}{a}\(\rho + p\).
\end{equation}

The Friedmann equation \eqref{eq_friedmann_density} shows that we can uniquely determine $\rho$ from $a$. The acceleration equation determines $p$ from both $a$ and $\rho$. Hence, the function $a$ determines uniquely both $\rho$ and $p$.

From the recent observations on supernovae, we know that the expansion is accelerated, corresponding to the existence of a positive cosmological constant $\Lambda$ \cite{RIE98,PER99}. The Friedmann's equations were expressed here without $\Lambda$, but this doesn't reduce the generality, because the equations containing the cosmological constant are equivalent to those without it, by the substitution
\begin{equation}
\label{eq_friedmann_lambda}
\begin{array}{l}
\bigg\{
\begin{array}{lll}
\rho &\to& \rho + \kappa^{-1}\Lambda  \\
p &\to& p - \kappa^{-1}\Lambda.  \\
\end{array}
\\
\end{array}
\end{equation}
Therefore, for simplicity we will continue to ignore $\Lambda$ in the following.

\image{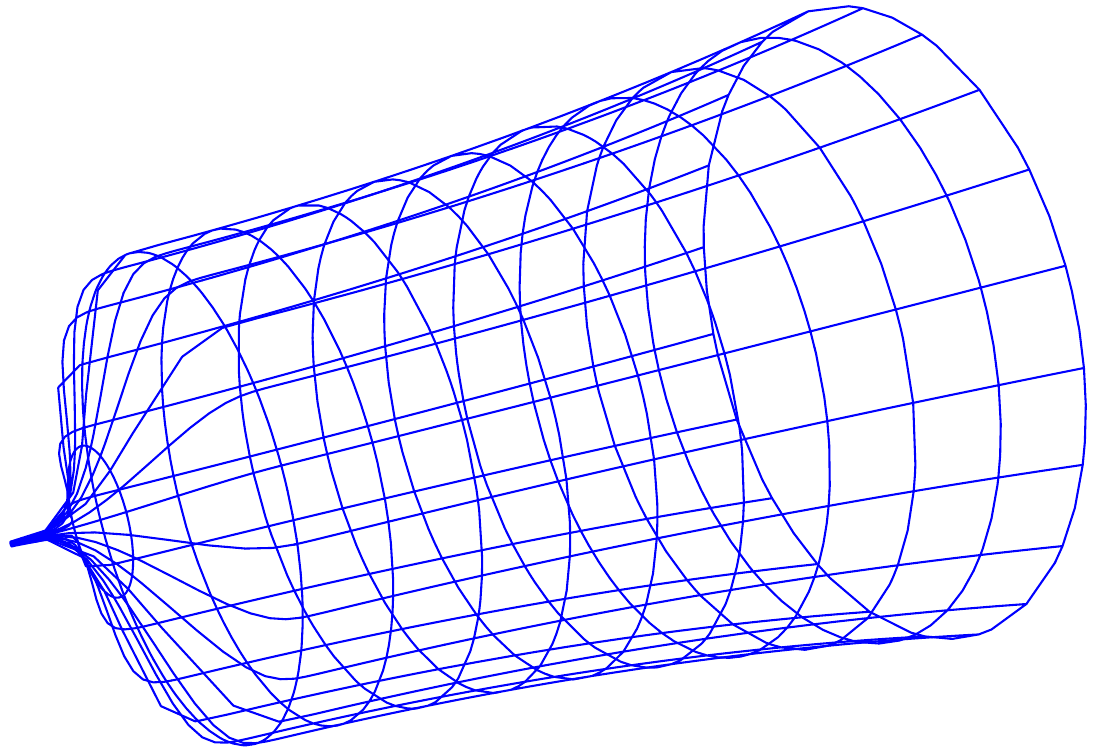}{0.5}{The standard view is that the universe originated from a very dense state, probably a singularity, and expanded, with a short period of very high acceleration (the inflation).}

The current standard view in cosmology is that the universe started with the Big Bang, which is in general assumed to be singular, and then expanded, with a very short period of exponentially accelerated expansion, called {\em inflation} (Fig. \ref{flrw-std}).

\section{The main ideas}

The solution proposed here is simple: to show that the singularities of the {\FLRW} model don't break the evolution equation, we show that the equations can be written in an equivalent form which avoids the infinities in a natural and invariant way. We consider useful to prepare the reader with some simple mathematical observations, which will clarify our proof. These observations can be easily understood, and combined they help us understand the Big Bang singularity in the {\FLRW} spacetime.

\subsection{Distance separation \vs topological separation}


Let's consider, in the space $\R^3$ parametrized by the coordinates $(u,v,w)$, the cylinder defined by the equation $v^2+w^2=1$. The transformation 
\begin{equation}
\label{eq_desing_cone}
\begin{array}{l}
\Bigg\{
\begin{array}{ll}
	x&=u \\
	y&=uv \\
	z&=uw \\
\end{array}
\\
\end{array}
\end{equation}
makes it into a cone in the space parametrized by $(x,y,z)$, defined by
\begin{equation}
\label{eq_cone}
	x^2-y^2-z^2=0.
\end{equation}

\image{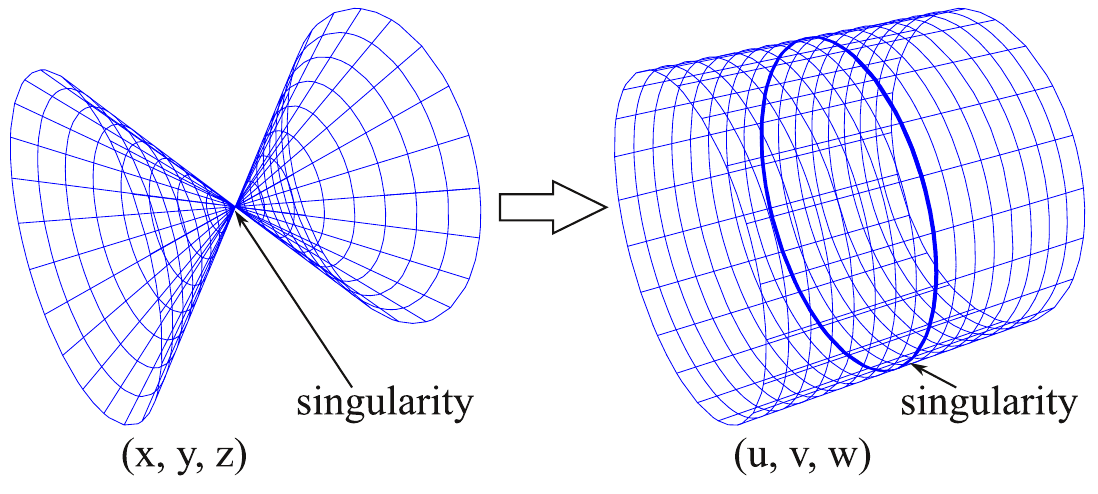}{0.75}{A cylinder may have the metric of a cone, but from topological viewpoint, it still remains a cylinder. Similarly, it is not necessary to assume that, at the Big Bang singularity, the entire space was a point, but only that the space metric was degenerate.}

The natural metric on the space $(x,y,z)$ induces, by pull-back, a metric on the cylinder $v^2+w^2=1$ from the space $(u,v,w)$. The induced metric on the cylinder is singular: the distance between any pair of points of the circle determined by the equations $u=0$ and $v^2+w^2=1$ is zero. But the points of that circle are distinct.

From the viewpoint of the singularities in General Relativity, the main implication is that just because the distance between two points is $0$, it doesn't mean that the two points coincide. We can see something similar already in Special Relativity: the $4$-distance between two events separated by a lightlike interval is equal to $0$, but those events may be distinct.

In fact, as we have seen in section \sref{s_flrw_intro}, the underlying manifold is $I\times\Sigma$, so there is no way to conclude that the space $\Sigma_t=\{t\}\times\Sigma$ reduces at a point, just because $a(t)=0$. However, we tried to make this more explicit, for pedagogical reasons.

\subsection{Degenerate warped product and singularities}

The mathematics of General Relativity is a branch of differential geometry, called {\em {\semiriem} (or pseudo-Riemannian) geometry} (see \eg \cite{ONe83}). It is a generalization of the Riemannian geometry, to the case when the metric tensor is still {\nondeg}, but its signature is not positive. In this geometric framework are defined notions like contraction, Levi-Civita connection, covariant derivative, Riemann curvature, Ricci tensor, scalar curvature, Einstein tensor. These are the main ingredients of the theory of General Relativity \cite{HE95,ONe83,Wal84}.

The problem is that at singularities these main ingredients can't be defined, or become infinite. The perfection of {\semiriem} geometry is broken there, and by this, it is usually concluded that the same happens with General Relativity.

In \cite{Sto11a} we introduced a way to extend {\semiriem} geometry to the degenerate case. There is a previous approach \cite{Kup87b,Kup96}, which works for metric of constant signature, and relies on objects that are not invariant. Our need was to have a theory valid for variable signature (because the metric changes from being {\nondeg} to being degenerate), and which in addition allows us to define the Riemann, Ricci and scalar curvatures in an invariant way, and something like the covariant derivative for the differential forms and tensor fields which are of use in General Relativity. After developing this theory, introduced in \cite{Sto11a}, we generalized the notion of warped product to the degenerate case, providing by this a way to construct useful examples of singularities of this well behaved kind \cite{Sto11b}.

From the mathematics of degenerate warped products it followed that a warped product like that involved in a {\FLRW} metric (equation \ref{eq_flrw_metric}) has only singularities which are well behaved, and which allow the extension of General Relativity to those points. At these singularities, the Riemann curvature tensor $R_{abcd}$ is not singular, and it is smooth if $a$ is smooth. The Einstein equation can be replaced by a densitized version, which allows the continuation to the singular points and avoids the infinities.

\subsection{What happens if the density becomes infinite?}
\label{s_infinite_density}

In the Friedmann equations \eqref{eq_friedmann_density}, \eqref{eq_acceleration}, and \eqref{eq_fluid}, the variables are $a$, the mass/energy density $\rho$ and the pressure density $p$. When $a\to 0$, $\rho$ appears to tend to infinity, because a finite amount of matter occupies a volume equal to $0$. Similarly, the pressure density $p$ may become infinite. How can we rewrite the equations to avoid the infinities? As it will turn out, not only there is a solution to do this, but the quantities involved are actually the natural ones, rather than $\rho$ and $p$. As present in the equations, both $\rho$ and $p$ are scalar fields. They are the mass and pressure density (see equation \eqref{eq_friedmann_stress_energy}), as seen by an observer moving with $4$-velocity $u=\frac{\partial_t}{|\partial_t|}$, in an orthonormal frame. However, there is no orthonormal frame for $a=0$, since the metric is degenerate at such points. Because of this, it makes no sense to use $\rho$ and $p$, which can't be defined there.

On the other hand, to find the mass from the mass density, we don't integrate the scalar $\rho$, but the differential $3$-form $\rho\volspace$, where
\begin{equation}
\label{eq_dvolspacet}
	\volspace(t,x,y,z) := \sqrt{g_{\Sigma_t}}\de x\wedge\de y\wedge\de z = a^3\sqrt{g_{\Sigma}}\de x\wedge\de y\wedge\de z
\end{equation}
is the volume form of the manifold $(\Sigma_t,g_{\Sigma_t})$.
Since the typical space is the same for all moments of time $t$, $\det_3 g_{\Sigma}$ is constant.

The {\em volume element}, or the {\em volume form} is defined as
\begin{equation}
\label{eq_dvol}
	\vol := \sqrt{-g}\de t\wedge\de x\wedge\de y\wedge\de z = a^3\sqrt{g_{\Sigma}}\de t\wedge\de x\wedge\de y\wedge\de z.
\end{equation}
It follows that
\begin{equation}
\label{eq_dvolspacest}
	\volspace =\operatorname{i}_{\partial_t}\vol.
\end{equation}

The values $\rho$ and $p$ which appear in the Friedmann equations coincide with the components of the corresponding densities only in an orthonormal frame, where the determinant of the metric equals $-1$, and we can omit $\sqrt{-g}$. But when $a\to 0$, an orthonormal frame would become singular, because $\det g\to 0$. A frame in which the metric has the determinant $-1$ will necessarily be singular when $a(t)=0$. In a non-singular frame, $\det g$ has to be variable, as it is in the comoving coordinate system of the {\FLRW} model.

We shall see in Theorem \ref{thm_bb_sing_resolved} that, unlike the (frame-dependent) scalars $\rho$ and $p$, the differential forms $\rho\volspace$, $p\volspace$, $\rho\vol$, and $p\vol$ are smooth, even if $a$ vanishes.

\section{The Big Bang singularity resolution}
\label{s_bb_sing_resolved_generic}

As explained in section \sref{s_infinite_density}, we should account in the mass/energy density and the pressure density for the term $\sqrt{-g}$.

Consequently, we make the following substitution:
\begin{equation}
\label{eq_substitution_densities}
\begin{array}{l}
\bigg\{
\begin{array}{ll}
	\widetilde\rho = \rho \sqrt{-g} = \rho a^3 \sqrt{g_{\Sigma}} \\
	\widetilde p = p \sqrt{-g} = p a^3 \sqrt{g_{\Sigma}} \\
\end{array}
\\
\end{array}
\end{equation}

We have the following result:
\begin{theorem}
\label{thm_bb_sing_resolved}
If $a$ is a smooth function, then the densities $\widetilde\rho$, $\widetilde p$, and the densitized stress-energy tensor $T_{ab}\sqrt{-g}$ are smooth (and therefore nonsingular), even at moments $t_0$ when $a(t_0)=0$.
\end{theorem}
\begin{proof}
The Friedmann equation \eqref{eq_friedmann_density} becomes
\begin{equation}
\label{eq_friedmann_density_tilde}
\widetilde\rho = \dsfrac{3}{\kappa}a\(\dot a^2 + k\) \sqrt{g_{\Sigma}},
\end{equation}
from which it follows that if $a$ is a smooth function, $\widetilde\rho$ is smooth as well.

The acceleration equation \eqref{eq_acceleration} becomes
\begin{equation}
\label{eq_acceleration_tilde}
\widetilde\rho + 3\widetilde p = -\dsfrac{6}{\kappa}a^2\ddot{a} \sqrt{g_{\Sigma}},
\end{equation}
which shows that $\widetilde p$ is smooth too. Hence, for smooth $a$, both $\widetilde \rho$ and $\widetilde p$ are non-singular.


The four-velocity vector field is $u=\dsfrac{\partial}{\partial_t}$, which is a smooth unit timelike vector. The densitized stress-energy tensor becomes therefore
\begin{equation}
\label{eq_friedmann_stress_energy_densitized}
T_{ab}\sqrt{-g} = \(\widetilde\rho+\widetilde p\)u_a u_b + \widetilde p g_{ab},
\end{equation}
which is smooth, because $\widetilde\rho$ and $\widetilde p$ are smooth functions. 
\end{proof}

\begin{remark}
We can write now a smooth densitized version of the Einstein Equation:
\begin{equation}
\label{eq_einstein_idx:densitized}
	G_{ab}\sqrt{-g} + \Lambda g_{ab}\sqrt{-g} = \kappa T_{ab}\sqrt{-g}.
\end{equation}

This equation is obtained from the same Lagrangian as the Einstein Equation, since the Hilbert-Einstein Lagrangian density $R\sqrt{-g} - 2\Lambda\sqrt{-g}$ already contains $\sqrt{-g}$. Hence, we don't have to change General Relativity to obtain it. What we have to do is just to avoid dividing by $\sqrt{-g}$ when it may be zero.
\end{remark}

Since $\vol = \sqrt{-g}\de t\wedge\de x\wedge\de y\wedge\de z$ \eqref{eq_dvol}, we can rewrite the densitized version of the Einstein Equation:
\begin{equation}
\label{eq_einstein_idx:volume}
	G_{ab}\vol + \Lambda g_{ab}\vol = \kappa T_{ab}\vol,
\end{equation}
where
\begin{equation}
\label{eq_friedmann_stress_energy_volume}
T_{ab}\vol = \(\rho\vol+p\vol\)u_a u_b + p\vol g_{ab},
\end{equation}
and all terms are finite and smooth everywhere, including at the singularity.

\image{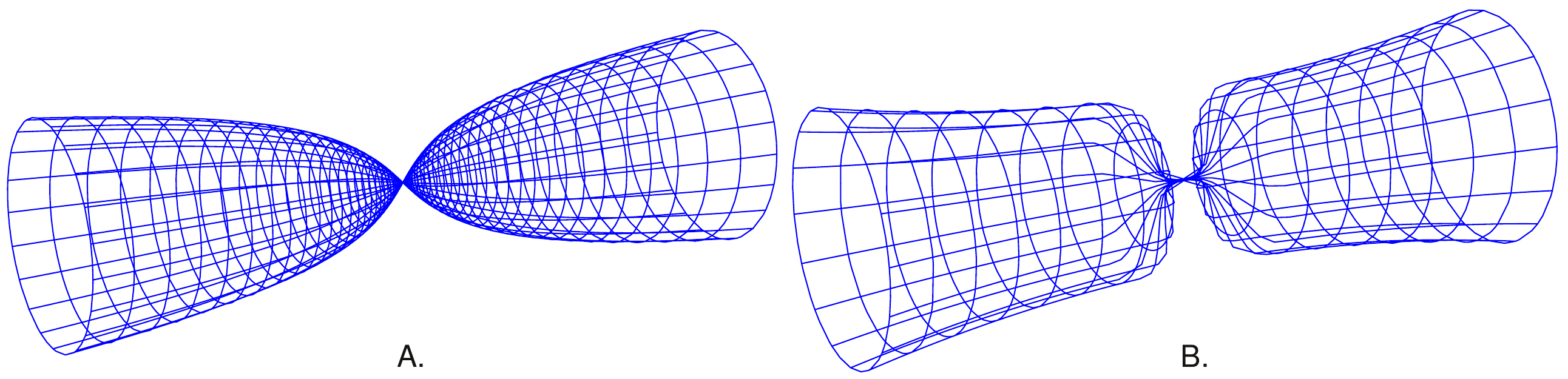}{1.0}{\textbf{A.} A schematic representation of a generic Big Bang singularity, corresponding to $a(0)=0$. The universe can be continued before the Big Bang without problems. \textbf{B.} A schematic representation of a Big Bang similar to an infinitesimal Big Bounce, corresponding to $a(0)=0$,  $\dot a(0)=0$,  $\ddot a(0)>0$.}

\begin{remark}
If $a(0)=0$, the equation \eqref{eq_friedmann_density_tilde} tells us that $\widetilde\rho(0)=0$. From these and equation \eqref{eq_acceleration_tilde} we see that $\widetilde p(0)=0$ as well. Of course, this doesn't necessarily tell us that $\rho$ or $p$ are zero at $t=0$, they may even be infinite. Figure \ref{flrw-def} A. shows how the universe will look, in general.
Another interesting possibility is that when $a(0)=0$, also $\dot a(0)=0$. In this case we may have $a(t)\geq 0$ around $t=0$, for example if $\ddot a(0)$, and obtain a Big Bang represented schematically in Fig. \ref{flrw-def} B.
This is very similar to a Big Bounce model, except that the singularity still appears.
\end{remark}

\section{Physical and geometric interpretation}

We compare the approach to the {\FLRW} singularity proposed here with the one we proposed in \cite{Sto12a}. An important advantage of the present approach is that the fundamental quantities, which remain smooth at the singularity, have a more natural physical interpretation.

The idea from \cite{Sto12a} is based not on the densitized version of the Einstein equation \eqref{eq_einstein_idx:densitized}, but on a version proposed in \cite{Sto12b}, named the \textit{expanded Einstein equation},
\begin{equation}
\label{eq_einstein_expanded}
	(G\circ g)_{abcd} + \Lambda (g\circ g)_{abcd} = \kappa (T\circ g)_{abcd}.
\end{equation}
In this equation we used the \textit{Kulkarni-Nomizu product} of two symmetric bilinear forms $h$ and $k$,
\begin{equation}
\label{eq_kulkarni_nomizu}
	(h\circ k)_{abcd} := h_{ac}k_{bd} - h_{ad}k_{bc} + h_{bd}k_{ac} - h_{bc}k_{ad}.
\end{equation}
The idea is based on the Ricci decomposition of the Riemann curvature tensor,
\begin{equation}
\label{eq_ricci_decomposition}
	R_{abcd} = S_{abcd} + E_{abcd} + C_{abcd},
\end{equation}
where $S_{abcd} = \frac{1}{24}R(g\circ g)_{abcd}$, $E_{abcd} = \frac{1}{2}(S \circ g)_{abcd}$, $S_{ab} := R_{ab} - \dsfrac{1}{4}Rg_{ab}$, and $C_{abcd}$ is the Weyl tensor.

The expanded Einstein equation \eqref{eq_einstein_expanded} can be rewritten explicitly as
\begin{equation}
\label{eq_einstein_expanded_explicit}
	2 E_{abcd} - 3 S_{abcd} + \Lambda (g\circ g)_{abcd} = \kappa (T\circ g)_{abcd},
\end{equation}

In \cite{Sto12a}, we showed that on a {\FLRW} spacetime, all the terms in equation \eqref{eq_einstein_expanded_explicit}, and consequently the expanded Einstein equation \eqref{eq_einstein_expanded}, are finite and smooth even at the singularity.
This suggests that the Ricci part of the Riemann curvature, although contains the same information as the Ricci tensor away from singularities, is more fundamental. Another advantage of {\quasireg} singularities is that, in dimension four, they satisfy automatically  Penrose's {\em Weyl curvature hypothesis} \cite{Pen79}, as we have shown in \cite{Sto12c}, while this is not known to be true for more general {\semireg} singularities.

The extension of the {\FLRW} solution through the singularity proposed in this article is isometric with the one proposed in \cite{Sto12a}. However, the solution proposed here has more advantages, because it is more general, and it leads to better physical interpretations.
It is expressed in terms of the densities $\rho \sqrt{-g}$ and $p \sqrt{-g}$, or equivalently $\rho \vol$ and $p \vol$, which are finite and smooth everywhere, including at the singularity. Unlike $\rho$ and $p$, which are usually called densities and in fact are scalars, and which are singular, $\rho \sqrt{-g}$ and $p \sqrt{-g}$ have the correct physical meaning of densities. For example, to obtain mass, one integrates the density. But on manifolds, one does not integrate scalars, but differential forms like $\rho \vol$ and $p \vol$. It is the differential form which has both the geometric meaning, and the physical meaning. This is why $\rho \vol$ and $p \vol$ make more sense than $\rho$ and $p$. In addition, they lead to the stress-energy tensor density $T_{ab}\sqrt{-g}$ from \eqref{eq_friedmann_stress_energy_densitized}, which is actually what we obtain from the matter Lagrangian, not the tensor $T_{ab}$. This is because the Lagrangian density is the density $\mc L_{\tn{matter}}\sqrt{-g}$, and not the scalar $\mc L_{\tn{matter}}$, and
\begin{equation}
\label{eq_stress_energy}
	T_{ab}\sqrt{-g} = 2\dsfrac{\delta\(\mc L_{\tn{matter}}\sqrt{-g}\)}{\delta g^{ab}}.
\end{equation}
In addition, the Lagrangian density from which the Einstein equation follows is, in terms of the Hilbert-Einstein Lagrangian density $R\sqrt{-g}$, the matter Lagrangian density $\mc L_{\tn{matter}}\sqrt{-g}$, and the cosmological constant $\Lambda$, is
\begin{equation}
\label{eq_lagrangian}
	\dsfrac{1}{2\kappa}\(R\sqrt{-g} - 2\Lambda\sqrt{-g}\) + \mc L\sqrt{-g}.
\end{equation}
In the solution presented here, although the scalar curvature $R$ is singular at $a(t)\to 0$, all terms from the Lagrangian density are finite and smooth. In the process to obtain the Einstein equation, the densitized one \eqref{eq_einstein_idx:densitized} is obtained. As long as we are not sure that $\sqrt{-g}\neq 0$, it is not allowed to divide by it, and the densitized Einstein equation is the one one should use.

Hence, the approach presented here in terms of densities (which are but the components of differential forms) has some advantages over the one presented in \cite{Sto12a}, for being more general, and having better geometrical and physical meaning.

\section{Perspectives}

We have seen that, when expressed in terms of proper variables, the equations of the {\FLRW} spacetime are not singular. Is this a lucky coincidence, or it reflects something more general?
Here we argue that this is a particular case of a more general situation, and this result is part of a larger series.

In \cite{Sto11a,Sto11b}, it is developed the geometry of manifolds endowed with metrics which are not necessarily non-degenerate everywhere. A special type of singularities, named {\em \semireg}, are shown to have nice properties. They admit covariant derivatives or lower covariant derivatives for the fields that are important, and the Riemann {\em curvature tensor} $R_{abcd}$ is smooth. They satisfy a {\em densitized Einstein equation}.

In the case of stationary {\em black holes}, the singularity $r=0$ is due to a combination of the fact that the coordinates are singular (just like in the case of the event horizon), and the metric is degenerate. In \cite{Sto11e,Sto11f,Sto11g} it is shown how we can remove the coordinate singularity, making $g_{ab}$ finite, and {\em analytic} at $r=0$.
In \cite{Sto12e} it is shown that black hole singularities are compatible with {\em global hyperbolicity}, which is required to restore the {\em conservation of information}.

In \cite{Sto12b}, a class of {\semireg} singularities, having good behavior is identified. They are named named {\em \quasireg}.
	In \cite{Sto12c}, it is shown that {\quasireg} singularities satisfy Penrose's {\em Weyl curvature hypothesis} \cite{Pen79}. A large class of {\em cosmological models} with big-bang that is not necessarily isotropic and homogeneous is identified.
	In \cite{Sto12d,Sto12f} is shown that {\semireg} and {\quasireg} singularities are accompanied by {\em dimensional reduction}, and connections with various approaches to perturbative {\em quantum gravity} are presented.

\subsection*{Acknowledgments}

The author thanks the reviewers for valuable comments and suggestions to improve the clarity and the quality of this paper.


\end{document}